\documentclass{IEEEtran}

\usepackage{subfigure}
\usepackage{longtable}
\usepackage{amsmath}
\usepackage{graphicx, amssymb}
\usepackage[dvips]{epsfig}
\usepackage{color}
\usepackage{cite}

\newtheorem{theorem}{Theorem}
\newtheorem{lem}{Lemma}
\newtheorem{remark}{Remark}
\newtheorem{asum}{Assumption}
\newcommand{\proof}{\noindent \textbf{Proof}: }

\begin{document}

\title{  Matched Disturbance Rejection for Port-Hamiltonian Systems with Coupled Dynamics }

\author{ M. Reza J. Harandi, et al.
}
\maketitle

	\begin{abstract}
This paper investigates the rejection of matched disturbances generated by coupled port-Hamiltonian (PH) dynamics in previously stabilized PH systems. The disturbance dynamics are incorporated into a unified PH representation, allowing the disturbance to affect the plant through both the matched input channel and an interconnection structure. Unlike existing results that typically impose restrictive structures on the disturbance dynamics, the proposed framework accommodates a more general class of coupled PH disturbances, including nonzero interconnection and damping terms. A baseline disturbance rejection scheme is first established for known disturbance storage parameters. The framework is then extended to the case of an unknown symmetric storage matrix through online parameter estimation. Two control designs are developed under different structural conditions, with the latter relaxing the dimensional restriction imposed by the first design. The proposed methods guarantee asymptotic convergence of the plant state to the desired equilibrium, while preserving a port-Hamiltonian representation of the closed-loop dynamics under the corresponding conditions. The results generalize existing disturbance rejection approaches and broaden their applicability to coupled PH disturbance models.
Experimental results on a ... validate the effectiveness of the proposed method.

	\end{abstract}
	
\section{Introduction}\label{s1}
Disturbance rejection is a fundamental problem in control engineering, particularly when external disturbances interact with the system dynamics in a structured manner. In many practical systems, disturbances cannot be adequately described as arbitrary exogenous signals or as bounded signals entering only through the plant input. Instead, they may possess their own internal dynamics and interact with the plant through physical or dynamic coupling. Such disturbances arise, for example, when the external excitation is generated by another dynamical subsystem or when the disturbance and plant exchange energy through an interconnection mechanism~\cite{park2022simultaneous,li2024coupling}. In these situations, the disturbance dynamics may contain dissipative, conservative, or temporarily energy-generating components, causing the disturbance amplitude to increase significantly over finite time intervals~\cite{li2024separation}. Thus, the disturbance state is not necessarily required to satisfy a prescribed uniform bound, although such transient growth does not imply that the disturbance diverges indefinitely. These characteristics make disturbance rejection substantially more challenging than the conventional matched-disturbance problem, where the disturbance is treated as an unknown signal entering solely through the input channel.

Port-Hamiltonian (PH) systems provide a natural framework for modeling and control of dynamical systems in which energy storage, dissipation, and interconnection mechanisms play a central role~\cite{yaghmaei2023contractive}. This framework has motivated several energy-based control methodologies, including interconnection and damping assignment passivity-based control (IDA-PBC) and proportional-integral-derivative passivity-based control (PID-PBC)~\cite{harandi2023reformulation,ortega2021pid}. Its explicit representation of energy and interconnection structures is also attractive for robust control, where stability and disturbance rejection can be analyzed using energy-based tools while preserving the physical structure of the closed-loop system~\cite{romero2025robust}. However, disturbance rejection becomes more involved when the disturbance itself is generated by a dynamical system and is coupled with the plant, since the resulting interaction introduces additional interconnection effects beyond conventional matched inputs.

Robust regulation and disturbance rejection have been extensively studied for PH systems. Early results considered robust integral control in the presence of constant input disturbances, including non-passive outputs and unmatched disturbances~\cite{ortega2012robust}. Robust regulation has also been investigated for infinite-dimensional PH systems, including boundary-controlled distributed-parameter systems~\cite{8023782}. More recently, integral IDA-PBC and related passivity-based methods have addressed constant and time-varying matched and unmatched disturbances in underactuated mechanical systems, including settings with unmeasured actuator dynamics~\cite{franco2025integral}. While these results demonstrate the effectiveness of PH-based approaches for robust regulation, the disturbances are generally modeled as exogenous signals rather than dynamical systems with their own internal states and energy structure. They therefore do not explicitly address disturbances whose dynamics are coupled with those of the plant.

A more closely related line of research considers disturbances generated by dynamical systems and incorporates their internal dynamics into the control design. In particular, Ferguson et al.~\cite{ferguson2019matched,ferguson2020matched} investigated matched disturbances generated by dynamical systems and developed PH-based disturbance rejection strategies that exploit their internal structure. However, their formulation assumes no dynamic interconnection between the plant and disturbance and restricts the disturbance dynamics to a purely skew-symmetric structure. The sinusoidal disturbance with unknown frequency considered in~\cite{ferguson2020matched} further represents a specialized case with additional structural restrictions on the disturbance model. Thus, these results do not directly cover a general setting in which the disturbance has its own dissipative and interconnection effects, may exhibit temporary energy growth, and is dynamically coupled with the plant. This motivates a more general framework for matched disturbance rejection that explicitly incorporates coupled disturbance dynamics and uncertainty in their energy-storage parameters.

In this paper, we investigate robust rejection of matched disturbances generated by coupled PH dynamics. The plant and disturbance are described within a unified PH framework with a separable total Hamiltonian, while their states are coupled through an interconnection matrix. The proposed framework allows the disturbance dynamics to contain both interconnection and dissipation effects and does not require the disturbance state to satisfy a prescribed uniform bound or converge to zero. Instead, the disturbance may exhibit transient growth over finite time intervals, while the plant state is guaranteed to asymptotically converge to the desired equilibrium. We first establish a baseline result for a known disturbance storage matrix and then develop two extensions that estimate an unknown storage matrix online under different structural conditions. In the resulting formulations, the closed-loop dynamics admit PH representations under the corresponding assumptions, preserving the energy-based structure used for stability analysis.

The main contributions of this paper are summarized as follows:
\begin{itemize}
	\item A generalized PH framework is developed for rejecting matched disturbances generated by dynamically coupled systems, explicitly accounting for both the matched disturbance input and the interconnection between the plant and disturbance dynamics.
	
	\item A disturbance rejection strategy is established for a  general class of dissipative disturbance dynamics, allowing transient growth over finite time intervals without requiring the disturbance state to satisfy a prescribed uniform bound, while guaranteeing asymptotic convergence of the plant state.
	
	\item Two online estimation mechanisms are developed to relax the assumption of a known disturbance storage matrix, with stability guarantees established under different structural conditions on the disturbance dynamics.
	
	\item The proposed control and estimation schemes preserve, under the corresponding assumptions, a PH representation of the closed-loop dynamics, providing a unified energy-based framework for stability analysis.
\end{itemize}

\textbf{Notation}.
For a differentiable function $H:\mathbb{R}^n\to\mathbb{R}$, its gradient with respect to $x$ is denoted by $\nabla_x H$. When the function depends only on $x$, the subscript $x$ is omitted for simplicity, i.e., $\nabla H\equiv\nabla_x H$. For a symmetric matrix $A\in\mathbb{R}^{n\times n}$, the notation $A\succ0$ and $A\succeq0$ indicates that $A$ is positive definite and positive semi-definite, respectively. The weighted quadratic norm is defined as
$|x|_A^2:=x^\top A x.$
For a matrix $B\in\mathbb{R}^{n\times m}$, the term \emph{full rank} means that $B$ has maximal possible rank, i.e.,
$\operatorname{rank}(B)=\min\{n,m\}.$ The range space and null space of $B$ are denoted by $\mathcal{R}(B)$ and $\operatorname{null}(B)$, respectively.  Unless otherwise stated, the dimensions and functional dependencies of variables and matrices are specified when they are first introduced.   The estimation of $*$ is denoted by $\hat{*}$ and $\tilde{*}=\hat{*}-*$ is estimation error.
\section{Main Results}\label{s2}
Matched disturbances with internal dynamics arise naturally in systems where an external disturbance is generated by, or dynamically coupled to, another physical process. Examples include mechanical and electromechanical systems subject to oscillatory loads, flexible structures interacting with unmodeled or external vibration modes, and motion-control systems affected by disturbances generated by auxiliary dynamic subsystems. In such settings, the disturbance is not necessarily an arbitrary exogenous signal, but may evolve according to its own energy-related dynamics and interact with the plant through physical coupling mechanisms. From a port-Hamiltonian perspective, these interactions can be represented through interconnection structures, while the disturbance may simultaneously enter the plant through an input port. This motivates considering a unified model in which the disturbance possesses internal dynamics and is coupled to the plant, rather than treating it solely as an unknown external input. The framework developed in this paper is intended to capture this broader class of disturbance scenarios while remaining sufficiently general to accommodate different physical realizations and disturbance-generation mechanisms.

\subsection{Problem formulation}
Consider the following port-Hamiltonian system
\begin{equation}\label{dyn}
	\begin{array}{c}
		\begin{bmatrix}
			\dot{x} \\ \dot{\omega}
		\end{bmatrix}
=\begin{bmatrix}
	F_1(x) & J(x) \\ -J^\top(x) & F_2(x)
\end{bmatrix}\begin{bmatrix}
\nabla_x H(x,\omega) \\ \nabla_\omega H(x,\omega)
\end{bmatrix}
 +\begin{bmatrix}
 	G(x)(u-d) \\ 0
 \end{bmatrix} \\
F_1(x):=J_1(x)-R_1(x),\qquad F_2:=J_2(x)-R_2(x) \\
H(x,\omega)=H_1(x)+H_2(\omega)=H_1(x)+\frac{1}{2}\|\omega\|_{L_m}^2\\
d=G_\omega^\top\nabla_\omega H_{e}(\omega)
	\end{array}
\end{equation}
where $x\in\mathbb{R}^n$, $u\in\mathbb{R}^m$, $d\in\mathbb{R}^m$, and $\omega\in\mathbb{R}^{n'}$ denote the system state, control input, matched disturbance, and the state of the disturbance dynamics, respectively. Furthermore, $H_1(x)$ is the Hamiltonian function of the plant and attains its minimum at the desired equilibrium $x^*$. The matrix $G(x)\in\mathbb{R}^{n\times m}$ denotes the full column rank input mapping matrix, while $J_1(x)\in\mathbb{R}^{n\times n}$ and $R_1(x)\in\mathbb{R}^{n\times n}$ are skew-symmetric and positive semi-definite matrices, respectively. Moreover, $J(x)\in\mathbb{R}^{n\times n'}$ denotes the interconnection matrix between the plant and the disturbance dynamics, $J_2(x)=-J_2^\top(x)\in\mathbb{R}^{n'\times n'}$ and $R_2(x)\succeq0$. Besides, $L_m^\top=L_m\succ0$, and the full rank matrix $G_\omega\in\mathbb{R}^{n'\times m}$ are constant.
The plant and the disturbance dynamics are represented within a unified PH framework. The disturbance affects the plant both through the matched input channel and through the interconnection structure, while the overall Hamiltonian remains separable as $H(x,\omega)=H_1(x)+H_2(\omega)$.
Throughout the paper, the unperturbed dynamics
\[
\dot{x}=F_1(x)\nabla H_1(x)
\]
are assumed to represent the closed-loop dynamics of a previously stabilized system obtained using an energy-based controller, such as IDA-PBC~\cite{harandi2023reformulation} or PID-PBC~\cite{ortega2021pid}. Consequently, the desired equilibrium $x^*$ is assumed to be asymptotically stable in the absence of disturbances.
\begin{asum}\label{as1}
	Consider the unperturbed system (\ref{dyn}) with $u=d=0$, i.e.,
	\[
	\dot{x}=F_1(x)\nabla H_1(x).
	\]
	Then, the equilibrium point $x^*$ is (globally) asymptotically stable.\hfill $\square$
\end{asum}
	A sufficient condition for this assumption is obtained by selecting
	$H_1(x)$ as a Lyapunov function, whose time derivative satisfies
	\begin{align*}
		\dot{H}_1
		=
		(\nabla H_1)^\top
		F_1
		\nabla H_1
		=
		-(\nabla H_1)^\top
		R_1
		\nabla H_1.
	\end{align*}
	Accordingly, Assumption~\ref{as1} holds if either $R_1$ is positive definite or the passive output $G^\top\nabla H_1$ is detectable~\cite{franco2025integral}.

We first consider the disturbance rejection problem for system~(\ref{dyn}) assuming that the disturbance dynamics are completely known. This result serves as the baseline for the subsequent developments, where the assumption of complete model knowledge is progressively relaxed. The objective is to design a control law that guarantees asymptotic stability of the equilibrium point $x^*$ in the presence of the matched disturbance.
\subsection{Known disturbance storage matrix}
\begin{lem}\label{lm1}
Consider system~(\ref{dyn}) under \textit{Assumption~\ref{as1}}. The controller and disturbance estimator are
		\begin{subequations}\label{con}
		\begin{align}
	u&=G_\omega^\top {L}_m\hat{\omega} \label{u}\\
	\dot{\hat{\omega}}&=F_2{L}_m\hat{\omega}-J^\top\nabla_x H-G_\omega G^\top\nabla_x H \label{om}. 
			\end{align}
\end{subequations}
 Then, the equilibrium point $x^*$ of the closed-loop system is (globally) asymptotically stable. $\hfill\square$
\end{lem}

\proof
The proof follows by showing that the resulting closed-loop dynamics admit a PH representation
which enables the stability analysis to be carried out within the standard energy-based framework
 in the following form
\begin{align}\label{cl}
	\begin{bmatrix}
		\dot{x} \\ \dot{{\omega}} \\ \dot{\tilde{\omega}}
	\end{bmatrix}
	=\begin{bmatrix}
		F_1(x) & J  & GG_\omega^\top \\
		-J^\top & F_2 & 0 \\
		-G_\omega G^\top & 0 & F_2 \\
	\end{bmatrix}
	\begin{bmatrix}
		\nabla_x \bar{H} \\ \nabla_{{\omega}}\bar{H}\\ \nabla_{\tilde{\omega}}\bar{H}
	\end{bmatrix},
\end{align}
where 
\begin{align*}
	\bar{H}(x,{\omega},\tilde{\omega})=H_1(x)+H_2(\omega)+\frac{1}{2}\|\tilde{\omega}\|_{L_m}^2,
\end{align*}
and $\tilde{\omega}=\hat{\omega}-\omega$.
Substituting (\ref{u1}) into the plant dynamics yields
\begin{align*}
	\dot{x}&=F_1(x)\nabla_x H +J\nabla_\omega H +GG_\omega^\top {L}_m\hat{\omega}-GG_\omega^\top L_m\omega\\&=
		\nabla_x \bar{H}+ J	\nabla_\omega
		 \bar{H}+GG_\omega^\top\nabla_{\tilde{\omega}}\bar{H}.
\end{align*}
The estimation error dynamics are obtained from (\ref{dyn}) and (\ref{om}) as
\begin{align*}
\dot{\tilde{\omega}}&=\dot{\hat{\omega}}-\dot{{\omega}}=	F_2{L}_m\hat{\omega}-J^\top\nabla_x H-G_\omega G^\top\nabla_x H +J^\top\nabla_x H\\&-F_2 L_m\omega =-G_\omega G^\top\nabla_x \bar{H} +F_2 L_m\tilde{\omega} =-G_\omega G^\top\nabla_x \bar{H} \\&+F_2\nabla_{\tilde{\omega}}\bar{H}.
\end{align*}
The time derivative of $\bar{H}$ along the closed-loop trajectories is
 \begin{align*}
  \dot{\bar{H}}&=(\nabla_x \bar{H})^\top F_1\nabla_x \bar{H}+(\nabla_\omega \bar{H})^\top F_2\nabla_\omega \bar{H}\\&+(\nabla_{\tilde{\omega}} \bar{H})^\top F_2\nabla_{\tilde{\omega}} \bar{H}\leq -(\nabla_x \bar{H})^\top R_1\nabla_x \bar{H}.
\end{align*}
The desired stability result follows directly from Assumption~\ref{as1}. $\hfill\blacksquare$
   
The closed-loop dynamics are naturally expressed in terms of the variables
$(x,\omega,\hat{\omega})$. By introducing the diffeomorphic coordinate
transformation
$\tilde{\omega}=\hat{\omega}-\omega$,
the closed-loop system can be equivalently represented in the PH form
(\ref{cl}).

\textit{Lemma~\ref{lm1}} establishes a basic disturbance rejection scheme under the assumption that the disturbance dynamics are dissipative. Since the primary objective there is stabilization, no damping is injected into the estimator dynamics. This observation motivates the following extension, where damping injection is incorporated to relax the dissipativity requirement on $R_2$ while preserving the PH structure of the closed-loop system. The following assumption is first introduced.
\begin{asum}\label{as2}
\textbf{a)} The interconnection matrix satisfies $J(x)=G(x)J_n(x)$ where $J_n\in\mathbb{R}^{m\times n'}$ is an arbitrary full rank matrix. 

\textbf{b)} $R_2$ is in the following form
	$$R_2=R_{2_n}-G_\omega \mathfrak{R}G_\omega^\top,$$
where $R_{2_n}$ is positive semi-definite, whereas
$\mathfrak{R}\in\mathbb{R}^{m\times m}$ is not necessarily negative semi-definite and may even be positive definite.
The motivation and interpretation of this decomposition are discussed in Remark~\ref{re2}. \hfill $\square$
\end{asum}
\begin{theorem}\label{th1}
Consider system~(\ref{dyn}) under Assumptions~\ref{as1} and~\ref{as2}. The controller and disturbance estimator are given by
			\begin{subequations}\label{con1}
		\begin{align}
			u&=G_\omega^\top {L}_m\hat{\omega}-J_n L_m\hat{\omega} \label{u1}\\
			{\hat{\omega}}&=x_a-G_\omega G^\top K x  \label{om1}\\
			\dot{x}_a&= F_2 L_m\hat{\omega}-G_\omega G^\top\nabla_x H\nonumber\\&+G_\omega G^\top KF_1\nabla_x H+G_\omega \dot{G}^\top K x,\nonumber 
		\end{align}
	\end{subequations}
	where $K\in\mathbb{R}^{n\times n}$ is a positive definite gain matrix, and $\dot{G}=dG/dt$. Then, the equilibrium point $x^*$ is (globally) asymptotically stable provided that
	\begin{align}\label{k}
		R_2-\operatorname{sym}\left\{
		G_\omega G^\top K(GJ_n-GG_\omega^\top)
		\right\}\succeq0,
	\end{align}
	where
	\[
	\operatorname{sym}(A):=\frac{1}{2}(A+A^\top),
	\]
	denotes the symmetric part of a square matrix $A$.
 $\hfill\square$
\end{theorem}
A sufficient condition for satisfaction of (\ref{k}) is $J_n=\kappa G_\omega^\top$ with $\kappa\in\mathbb{R}$ and 
\begin{align*}
	\lambda_{min}\{R_2+G_\omega G^\top KGG_\omega^\top\}\geq\lambda_{max}\{G_\omega (\mathfrak{R} +\kappa G^\top KG)G_\omega^\top\}.
\end{align*}
Moreover, global asymptotic stability is guaranteed provided that Assumption~\ref{as1} holds globally and $H_1$ is radially unbounded.$\hfill\square$

\proof
To prove \textit{Theorem~\ref{th1}}, first consider the closed-loop dynamics in the error coordinates $(x,\tilde{\omega})$. These dynamics can be written as
\begin{align}\label{cl1}
	\begin{bmatrix}
		\dot{x}  \\ \dot{\tilde{\omega}}
	\end{bmatrix}
	=\begin{bmatrix}
		F_1 & GG_\omega^\top-J \\
		J^\top-G_\omega G^\top & F_2-F_\omega  \\
	\end{bmatrix}
	\begin{bmatrix}
		\nabla_x \bar{H} \\  \nabla_{\tilde{\omega}}\bar{H}
	\end{bmatrix},
\end{align}
where 
\begin{align*}
	\bar{H}(x,\tilde{\omega})=H_1(x)+\frac{1}{2}\|\tilde{\omega}\|_{L_m}^2,
\end{align*}
and
\begin{align}\label{r}
	F_\omega(x):=G_\omega G^\top K (G G_\omega^\top-J).
\end{align}
Under Assumption~\ref{as2},
\[
J=GJ_n,
\qquad
F_2=J_2-R_{2_n}+G_\omega\mathfrak{R}G_\omega^\top.
\]
Substituting the control law (\ref{u1}) into (\ref{dyn}) yields
\begin{align*}
		\dot{x}&=F_1\nabla_x H +GJ_nL_m \omega +GG_\omega^\top {L}_m\hat{\omega}-GJ_n L_m\hat{\omega}\\&-GG_\omega^\top L_m\omega=
F_1	\nabla_x \bar{H}- GJ_n L_m\tilde{\omega}+GG_\omega^\top L_m\tilde{\omega}\\&=F_1\nabla_x \bar{H}+(-J+GG_\omega^\top)
	\nabla_{\tilde{\omega}}\bar{H}.
\end{align*}
Using (\ref{om1}) and the system dynamics in (\ref{dyn}), the estimation error dynamics are obtained as
\begin{align*}
	\dot{\tilde{\omega}}&=\dot{\hat{\omega}}-\dot{{\omega}}=	F_2 L_m\hat{\omega}-G_\omega G^\top\nabla_x H+G_\omega G^\top KF_1\nabla_x H\\&+G_\omega \dot{G}^\top K x-G_\omega \dot{G}^\top K x-G_\omega G^\top K \dot{x}+J^\top\nabla_x H\\&-F_2 L_m\omega
	=J^\top\nabla_x H-G_\omega G^\top\nabla_x H+F_2L_m\tilde{\omega}\\&+G_\omega G^\top KF_1\nabla_x H -G_\omega G^\top K(F_1\nabla_x H-JL_m\tilde{\omega}\\&+GG_\omega^\top L_m\tilde{\omega})
	=(J^\top-G_\omega G^\top)\nabla_x \bar{H}+(F_2\\&+G_\omega G^\top K J-G_\omega G^\top KG G_\omega^\top)\nabla_{\tilde{\omega}}\bar{H},
\end{align*}
which is precisely the second row of (\ref{cl1}). Consider $\bar{H}$ as a Lyapunov function. Its time derivative along the closed-loop trajectories is given by
\begin{align}
	\dot{\bar{H}}&=(\nabla_x \bar{H})^\top F_1\nabla_x \bar{H}-(\nabla_{\tilde{\omega}} \bar{H})^\top\big(R_2-\operatorname{sym}\{G_\omega G^\top K \nonumber\\&(GJ_n-G G_\omega^\top)\}\big)\nabla_{\tilde{\omega}} \bar{H}\leq -(\nabla_x \bar{H})^\top R_1\nabla_x \bar{H},\label{ly}
\end{align}
where (\ref{k}) was used. Therefore, $\bar{H}$ is nonincreasing along the closed-loop trajectories. Moreover, since $\bar{H}$ is positive definite with respect to $(x^*,\tilde{\omega})=(x^*,0)$, and \textit{Assumption~\ref{as1}} guarantees asymptotic stability of $x^*$ for the unperturbed plant dynamics, the state $x$ asymptotically converges to the desired equilibrium point $x^*$.$\hfill\blacksquare$

Here, the objective is asymptotic regulation of the plant state rather than boundedness or regulation of the disturbance generator state. Besides,
by (\ref{ly}) and LaSalle's invariance principle, the term
$G_\omega^\top L_m\tilde{\omega}$ converges to zero provided that
\begin{align*}
	G_\omega \mathfrak{R}G_\omega^\top-\operatorname{sym}\left\{G_\omega G^\top K (GJ_n-G G_\omega^\top)\right\}\succ0.
\end{align*}
\begin{remark}\label{re2}
	The controller proposed in \textit{Theorem~\ref{th1}} is capable of rejecting matched disturbances whose dynamics may generate unbounded trajectories. Since the control input in (\ref{u1}) explicitly depends on the estimated disturbance state $\omega$, the control effort may also become unbounded as $d\rightarrow\infty$. Such disturbances are rarely encountered in practical engineering systems, where external disturbances are typically bounded or exhibit transient growth over finite time intervals. For instance, severe wind gusts acting on unmanned aerial vehicles or wind turbines may cause the disturbance amplitude to increase rapidly over a limited period before eventually returning to a bounded regime. Such transient behavior can, in principle, be represented by allowing the effective matrix $\mathfrak R$ to vary with time, such that some of its eigenvalues become temporarily positive over finite time intervals while remaining non-positive otherwise.
\end{remark}
\subsection{Unknown disturbance storage matrix}
\textit{Theorem~\ref{th1}} extends the result of \textit{Lemma~\ref{lm1}} by relaxing the condition imposed on $R_2$ and introducing damping injection into the estimator dynamics, thereby allowing the rejection of matched disturbances whose dynamics may generate unbounded trajectories. Next, the assumption that $L_m$ is known is relaxed by estimating it online. The resulting control law and estimation scheme are presented in the following theorem.
\begin{theorem}\label{th2}
Consider a PH system subject to a matched disturbance with coupled dynamics given by (\ref{dyn}). Suppose that Assumption~\ref{as1} holds, $J=GJ_n$, and $L_m$ is unknown. Furthermore, assume that
\[
(G_\omega-J_n^\top)(G_\omega^\top-J_n)
\]
is invertible.
The control law and the estimators for $\omega$ and $L_m$ are given by
	 		\begin{subequations}\label{con2}
	 	\begin{align}
			u&=G_\omega^\top \hat{L}_m\hat{\omega}-J_n 
			\hat{L}_m\hat{\omega} \label{u2}\\
\dot{\hat{\omega}}&= F_2 \hat{L}_m\hat{\omega}-G_\omega G^\top\nabla_x H  \label{om2}\\
	 			\hat{l}&=x_a-\hat{\Omega}_m^\top F_2 ^\top \Lambda x  \label{l} \\
	 			\dot{x}_a&=\dot{\hat{\Omega}}_m^\top F_2 ^\top \Lambda x+{\hat{\Omega}}_m^\top \dot{F_2} ^\top \Lambda x
	 		-\hat{\Omega}^\top_m (G_\omega G^\top-J_n^\top G^\top)\nabla_x H\nonumber\\& +\hat{\Omega}^\top_m F_2^\top\Lambda F_1\nabla_x H,\nonumber
	 	\end{align}
	 \end{subequations}
	 where 
	 \begin{align*}
	 	\Lambda=\big((G_\omega-J_n^\top) (G_\omega ^\top -J_n)\big)^{-1}(G_\omega-J_n^\top) (G^\top G)^{-1}G^\top,
	 \end{align*}
	$\Omega_m(\omega)\in\mathbb{R}^{n'\times n'(n'+1)/2}$ and
	$l\in\mathbb{R}^{n'(n'+1)/2}$ are defined such that
	$\Omega_ml=L_m\omega$; see Appendix~\ref{app:parametrization}
	for the detailed parametrization. Moreover, $x_a$ is an auxiliary dynamic vector. Then, the state $x$ asymptotically converges to the desired equilibrium point $x^*$.
	 \hfill$\square$
\end{theorem}
 Note that invertibility of $(G_\omega-J_n^\top) (G_\omega ^\top -J_n)$ is equivalent with $n'\leq m$ since the matrices are full rank ($G_\omega\neq J_n^\top$).
 
 \proof
 First, we show that the closed-loop dynamics can be written as
 \begin{align}\label{cl2}
 	\begin{bmatrix}
 		\dot{x}  \\ \dot{\tilde{\omega}} \\ \dot{\tilde{l}}
 	\end{bmatrix}
 	=
 \mathfrak{F}(x)
 	\begin{bmatrix}
 		\nabla_x \bar{H} \\  \nabla_{\tilde{\omega}}\bar{H} \\ \nabla_{\tilde{l}}\bar{H}
 	\end{bmatrix},
 \end{align}
 where the augmented Hamiltonian is defined as
 \begin{align*}
 	\bar{H}(x,\tilde{\omega},\tilde{l})=H_1(x)+\frac{1}{2}\|\tilde{\omega}\|_{L_m}^2+\frac{1}{2}\|\tilde{l}\|^2,
 \end{align*}
 with $\tilde{l}=\hat{l}-l$,
  and $\mathfrak{F}(x)$ is
 \begin{align*}
 	\begin{bmatrix}
 		F_1 & GG_\omega^\top-J & (GG_\omega^\top-J )\hat{\Omega}_m\\
 		J^\top-G_\omega G^\top & F_2 & F_2\hat{\Omega} _m \\
 		\hat{\Omega}_m^\top (J^\top-G_\omega G^\top) & -\hat{\Omega}_m^\top F_2^\top & \hat{\Omega}_m^\top F_2^\top \hat{\Omega}_m
 	\end{bmatrix}.
 \end{align*}
 By defining $\tilde{L}_m=\hat{L}_m-L_m$,
the dynamics of $x$ obtained from (\ref{u2}) are given by
 \begin{align*}
 	\dot{x}&=F_1\nabla_x H +GJ_nL_m \omega +GG_\omega^\top \hat{L}_m\hat{\omega}-GJ_n \hat{L}_m\hat{\omega}\\&-GG_\omega^\top L_m\omega=
 	F_1	\nabla_x \bar{H}- GJ_n L_m\tilde{\omega}+GG_\omega^\top L_m\tilde{\omega}\\&- GJ_n\tilde{L}_m\hat{\omega}+GG_\omega^\top \tilde{L}_m\hat{\omega}
 	=F_1\nabla_x \bar{H}+(GG_\omega^\top-J) 
 	\nabla_{\tilde{\omega}}\bar{H}\\&+ (GG_\omega^\top-J )\hat{\Omega}_m\nabla_{\tilde{l}}\bar{H},
 \end{align*}
where $\pm (GG_\omega^\top-GJ_n)L_m\hat{\omega}$ is added to obtain the second equality, and the last equality follows from
$\tilde{L}_m\hat{\omega}=\hat{\Omega}_m\tilde{l}$. Using (\ref{om2}), the estimation error dynamics are obtained as
 \begin{align*}
 	\dot{\tilde{\omega}}&= F_2 \hat{L}_m\hat{\omega}-G_\omega G^\top\nabla_x H +J^\top\nabla_x H-F_2 L_m\omega=(J^\top\\&-G_\omega G^\top)\nabla_x H+F_2L_m\tilde{\omega}+F_2\tilde{L}_m\hat{\omega}=(J^\top-G_\omega G^\top)\nabla_x \bar{H}\\&+F_2\nabla_{\tilde{\omega}}\bar{H}+F_2\hat{\Omega}_m\nabla_{\tilde{l}}\bar{H},
 \end{align*}
 which is equal to second row of (\ref{cl2}). The dynamics of $\tilde{l}$ are obtained by differentiating (\ref{l}), which yields
 \begin{align*}
\dot{\tilde{l}}&=\dot{\hat{\Omega}}_m^\top F_2 ^\top \Lambda x+{\hat{\Omega}}_m^\top \dot{F}_2 ^\top \Lambda x
-\hat{\Omega}^\top_m (G_\omega G^\top-J_n^\top G^\top)\nabla_x H \\&+\hat{\Omega}^\top_m F_2^\top\Lambda F_1\nabla_x H 	-\dot{\hat{\Omega}}_m^\top F_2 ^\top \Lambda x
-{\hat{\Omega}}_m^\top \dot{F}_2 ^\top \Lambda x
-\hat{\Omega}_m^\top F_2 ^\top \Lambda \\&\big(F_1\nabla_x \bar{H}+(GG_\omega^\top-J) 
\nabla_{\tilde{\omega}}\bar{H}+ (GG_\omega^\top-J )\hat{\Omega}_m\nabla_{\tilde{l}}\bar{H}\big)=\\&-\hat{\Omega}^\top_m (G_\omega G^\top-J^\top )\nabla_x \bar{H}-\hat{\Omega}_m^\top F_2 ^\top\nabla_{\tilde{\omega}}\bar{H}-\hat{\Omega}_m^\top F_2 ^\top\hat{\Omega}_m\nabla_{\tilde{l}}\bar{H}.
 \end{align*} 
The time derivative of $\bar{H}$ along the closed-loop trajectories is given by
\begin{align}
	\dot{\bar{H}}
	&=
	(\nabla_x\bar{H})^\top F_1\nabla_x\bar{H}
	+
	(\nabla_{\tilde{\omega}}\bar{H})^\top
	F_2\nabla_{\tilde{\omega}}\bar{H}
	\nonumber\\
	&+
	(\nabla_{\tilde{l}}\bar{H})^\top
	\hat{\Omega}_m^\top F_2^\top\hat{\Omega}_m
	\nabla_{\tilde{l}}\bar{H}
\leq
	-(\nabla_x\bar{H})^\top R_1\nabla_x\bar{H}.
	\label{dot}
\end{align}
Therefore, by \textit{Assumption~\ref{as1}}, the state $x$ asymptotically converges to the desired equilibrium point $x^*$.$\hfill\blacksquare$
 
From (\ref{dot}), it follows that, in addition to $x-x^*$, the signals
$R_2L_m\tilde{\omega}$ and $R_2\hat{\Omega}_m\tilde{l}$ converge to zero.
Thus, if $R_2L_m$ and $R_2\hat{\Omega}_m$ are full column rank in a
neighborhood of $x^*$, then $\tilde{\omega}\to0$ and $\tilde{l}\to0$.
Otherwise, the estimation errors may converge to the corresponding null
spaces, i.e.,
\[
\tilde{\omega}\in\operatorname{null}(R_2L_m),
\qquad
\tilde{l}\in\operatorname{null}(R_2\hat{\Omega}_m),
\]
and convergence of the estimation errors to zero may be established by
applying LaSalle's invariance principle.
 
In \textit{Theorem~\ref{th2}}, the invertibility of
$(G_\omega-J_n^\top)(G_\omega^\top-J_n)$ was assumed, which implies
$n'\leq m$. This assumption may restrict the applicability of the proposed
approach. Moreover, in some cases, damping can be injected into the
dynamics of $\tilde{\omega}$. Motivated by these considerations, an
alternative control law is proposed under the following assumption.
\begin{asum}\label{as3}
Part \textbf{b)} of \textit{Assumption~\ref{as2}} holds. Furthermore, 
	 	\begin{subequations}\label{fer}
		\begin{align}
			& G_\omega  G_\omega ^\top F_2\hat{\Omega}_m=F_2\hat{\Omega}_m\label{fe}\\
			&G_\omega^\top \hat{\Omega}_m=0\label{gom}\\
			& J=G^{\bot^\top}\Gamma\label{j}\\
			&\Gamma\hat{\Omega}_m=0,\label{ga}
		\end{align}
	\end{subequations}
	where $G^\bot$ is left kernel of $G$, and $\Gamma(x)\in\mathbb{R}^{n-m\times n'}$.
	The interpretation and practical relevance of the structural conditions in \eqref{fer} are discussed following Theorem~\ref{th3}.\hfill $\square$
\end{asum}
The next result provides an alternative design that avoids the dimensional restriction introduced by Theorem 2, at the expense of imposing different structural conditions on the disturbance dynamics introduced in (\ref{fer}).
 \begin{theorem}\label{th3}
Consider system~(\ref{dyn}) and suppose that Assumptions~\ref{as1} and~\ref{as3} hold.
The control law and the estimators are given by
 	 		\begin{subequations}\label{con3}
	\begin{align}
		u&=G_\omega^\top \hat{L}_m\hat{\omega}-J_n 
		\hat{L}_m\hat{\omega} \label{u3}\\
		{\hat{\omega}}&=x_a-G_\omega G^\top K x  \label{om3}\\
		\dot{x}_a&= F_2 \hat{L}_m\hat{\omega}-G_\omega G^\top\nabla_x H+G_\omega G^\top KF_1\nabla_x H\nonumber\\&+G_\omega \dot{G}^\top K x,\nonumber \\
		\hat{l}&=y_a-\hat{\Omega}_m^\top  F_2^\top G_\omega (G^\top G)^{-1}G^\top x  \label{l1} \\
		\dot{y}_a&=\hat{\Omega}_m^\top F_2^\top G_\omega  d\{(G^\top G)^{-1}G^\top\}/dt\; x
		-\hat{\Omega}^\top_m (G_\omega G^\top\nonumber\\&-J_n^\top G^\top)\nabla_x H +\dot{\hat{\Omega}}_m^\top F_2^\top G_\omega (G^\top G)^{-1}G^\top x+{\hat{\Omega}}_m^\top\nonumber\\&  F_2^\top G_\omega (G^\top G)^{-1}G^\top F_1\nabla_x H+\hat{\Omega}_m^\top  \dot{F}_2^\top G_\omega (G^\top G)^{-1}G^\top x,\nonumber
	\end{align}
\end{subequations}
where $x_a$ and $y_a$ are auxiliary variables. Then, the state $x$ converges to $x^*$ if
	\begin{align}
	R_2-\operatorname{sym}\left\{G_\omega G^\top K (G^{\bot^\top}\Gamma-G G_\omega^\top)\right\}\succeq0.\label{kk}
\end{align}
  $\hfill\square$
 \end{theorem}

\proof
To prove Theorem~\ref{th3}, it suffices to show that the closed-loop dynamics can be written in the form of (\ref{cl2}), with
\begin{align}\label{f}
	\mathfrak{F}(x)=
	\begin{bmatrix}
		F_1 & GG_\omega^\top-J & 0\\
		J^\top-G_\omega G^\top & F_2-F_\omega & F_2\hat{\Omega}_m\\
		0 & -\hat{\Omega}_m^\top F_2^\top & 0
	\end{bmatrix},
\end{align}
where $F_\omega$ was defined in (\ref{r}), with $J=G^{\bot^\top}\Gamma$.
Since the control law in (\ref{u3}) is identical to that in (\ref{u2}),
the dynamics of $x$ is obtained in the same manner as in
Theorem~\ref{th2}. Note that due to (\ref{fer}), the term $(GG_\omega^\top-J )\hat{\Omega}_m$ vanishes. To derive the dynamics of $\tilde{\omega}$, differentiating (\ref{om3}) yields
  \begin{align*}
 	\dot{\tilde{\omega}}&= F_2 \hat{L}_m\hat{\omega}-G_\omega G^\top\nabla_x H+G_\omega G^\top KF_1\nabla_x H+G_\omega \dot{G}^\top K x\\&-G_\omega \dot{G}^\top K x-G_\omega {G}^\top K\big(F_1\nabla_x \bar{H}+(GG_\omega^\top-J) 
 	\nabla_{\tilde{\omega}}\bar{H}+\\& (GG_\omega^\top-J )\hat{\Omega}_m\nabla_{\tilde{l}}\bar{H}\big) +J^\top\nabla_x H-F_2 L_m\omega=(J^\top-G_\omega G^\top)\\&\nabla_x H+(F_2-G_\omega G^\top K (G G_\omega^\top-J))L_m\tilde{\omega}+(F_2-G_\omega G^\top K \\&(G G_\omega^\top-J))\tilde{L}\hat{\omega}=(J^\top-G_\omega G^\top)\nabla_x \bar{H}+(F_2-F_\omega)\nabla_{\tilde{\omega}}\bar{H}\\&+(F_2-F_\omega)\hat{\Omega}_m\nabla_{\tilde{l}}\bar{H}=(J^\top-G_\omega G^\top)\nabla_x \bar{H}+(F_2-F_\omega)\\&\nabla_{\tilde{\omega}}\bar{H}+F_2\hat{\Omega}_m\nabla_{\tilde{l}}\bar{H},
 \end{align*}
 where (\ref{gom}), (\ref{j}), and (\ref{ga}) are used in the last equality. The dynamics of $\tilde{l}$ is obtained by differentiating (\ref{l1}), which gives
 \begin{align*}
 	\dot{\tilde{l}}&=\hat{\Omega}_m^\top  F_2^\top G_\omega  d\{(G^\top G)^{-1}G^\top\}/dt\; x
 	-\hat{\Omega}^\top_m (G_\omega G^\top-J^\top )\\&\nabla_x H +\dot{\hat{\Omega}}_m^\top \Xi x+{\hat{\Omega}}_m^\top \Xi F_1\nabla_x H +\hat{\Omega}_m^\top  \dot{F}_2^\top G_\omega (G^\top G)^{-1}G^\top x\\&
 	-\dot{\hat{\Omega}}_m^\top \Xi x -\hat{\Omega}_m^\top   F_2^\top G_\omega  d\{(G^\top G)^{-1}G^\top\}/dt\; x
 	-\hat{\Omega}_m^\top  \dot{F}_2^\top G_\omega\\&  (G^\top G)^{-1}G^\top x
 	-\hat{\Omega}_m^\top\Xi \big(F_1\nabla_x \bar{H}+(GG_\omega^\top-J) 
 	\nabla_{\tilde{\omega}}\bar{H}\\&+ (GG_\omega^\top-J )\hat{\Omega}_m\nabla_{\tilde{l}}\bar{H}\big)
 	=  -\hat{\Omega}^\top_m F_2^\top G_\omega G_\omega^\top\nabla_{\tilde{\omega}} \bar{H}
 	\\&-\hat{\Omega}^\top_m F_2^\top G_\omega G_\omega^\top\hat{\Omega}_m\nabla_{\tilde{l}}\bar{H}
 \end{align*}
 where
 $$\Xi= F_2^\top G_\omega (G^\top G)^{-1}G^\top,$$
 and (\ref{fer}) is used to obtain the last equality. Using (\ref{fe}) and (\ref{gom}), the above expression simplifies to
 \begin{align*}
 \dot{\tilde{l}}= -\hat{\Omega}^\top_m F_2^\top \nabla_{\tilde{\omega}} \bar{H},
 \end{align*}
 which confirms (\ref{f}).Taking $\bar{H}$ as a Lyapunov function, its time derivative satisfies
\begin{align*}
	\dot{\bar{H}}
	&=(\nabla_x\bar{H})^\top F_1\nabla_x\bar{H}
	+(\nabla_{\tilde{\omega}}\bar{H})^\top
	(F_2-F_\omega)\nabla_{\tilde{\omega}}\bar{H}\\
	&\leq-(\nabla_x\bar{H})^\top R_1\nabla_x\bar{H},
\end{align*}
where the inequality follows from (\ref{kk}). Hence, asymptotic stability of
$x^*$ follows from Assumption~\ref{as1}.
\hfill$\blacksquare$

\textbf{Discussion:} In Theorem~\ref{th3}, a robust control law was proposed to reject a class of
matched disturbances with an unknown storage function, acting on a previously stabilized PH system. Now let us focus on conditions reported in (\ref{fer}). Let us first consider condition (\ref{fe}). Assume that the columns of
$G_\omega$ are orthonormal, i.e., $G_\omega^\top G_\omega=I_m$.
Then, $G_\omega G_\omega^\top$ is the orthogonal projector onto the range
space of $G_\omega$, denoted by $\mathcal{R}(G_\omega)$. Hence,
condition (\ref{fe}) holds if and only if
\[
\mathcal{R}(F_2\hat{\Omega}_m)\subseteq\mathcal{R}(G_\omega).
\]
Moreover, if $F_2$ is nonsingular, then
$\operatorname{rank}(F_2\hat{\Omega}_m)=\operatorname{rank}(\hat{\Omega}_m),$
which implies the necessary condition
$\operatorname{rank}(\hat{\Omega}_m)\le m$ 
since $m<n'$. Furthermore, from (\ref{gom}) it is deduced that $\operatorname{rank}(\hat{\Omega}_m)\le n'-m$. Hence,
$$\operatorname{rank}(\hat{\Omega}_m)\le \min\{n'-m,m\}.$$
Besides, condition (\ref{ga}) is satisfied by various choices of
$\Gamma$, including $\Gamma=0$ and
$\Gamma=\Psi G_\omega^\top$, where
$\Psi\in\mathbb{R}^{n-m\times m}$ is arbitrary.
A particular example satisfying Assumption~\ref{as3} is a sinusoidal
disturbance with unknown frequency, as studied in~\cite{ferguson2020matched} with $G_\omega^\top=[I_m,0], \Gamma=0$ and $n'=2m$.
This example demonstrates that \textit{Assumption~\ref{as3}} is not merely an abstract condition.
\begin{remark}\label{re1}
	The results presented in this paper provide a general framework for robust rejection of matched disturbances in port-Hamiltonian systems with coupled disturbance dynamics. In \textit{Theorem~\ref{th1}}, the disturbance is allowed to interact with the plant through both the input port and the interconnection structure, while the disturbance dynamics are not restricted to the special cases considered in previous studies. In particular, the disturbance state $\omega$ is not required to be bounded or to converge to zero and may exhibit unbounded trajectories, depending on the structure of $\mathfrak{R}$. Moreover, the proposed framework does not require the disturbance storage matrix $L_m$ to be known. This restriction is subsequently relaxed in Theorems~\ref{th2} and~\ref{th3}, where online estimation mechanisms are introduced to handle an unknown $L_m$ under different structural conditions on the disturbance dynamics.
	In comparison with~\cite{ferguson2019matched,ferguson2020matched}, the proposed framework considers a substantially more general class of disturbance dynamics. In~\cite{ferguson2020matched}, the disturbance dynamics are considered under the restrictions $J=0$ and a purely skew-symmetric $F_2$, and the subsequent sinusoidal disturbance model represents a particular special case with a more restrictive structure for the matrices corresponding to $L_m,F_2$ and $G_\omega$. In contrast, the present work allows a nonzero interconnection matrix $J$ and a general $F_2=J_2-R_2$, thereby accommodating both interconnection and damping effects in the disturbance dynamics. Furthermore, the matrix $L_m$ is not restricted to a diagonal structure and can be an unknown symmetric matrix, which substantially enlarges the class of admissible disturbance models.
	{Therefore, the main contribution of this work is not limited to extending a specific disturbance model or a particular sinusoidal case. Rather, it establishes a unified robust control framework comprising complementary control designs for a broader class of matched disturbances generated by coupled PH dynamics, while allowing unknown disturbance storage parameters and, in appropriate cases, preserving a PH representation of the closed-loop system.}
	\hfill$\square$
\end{remark}

\appendix
\section*{Parametrization of a Symmetric Matrix}
\label{app:parametrization}

Consider a symmetric matrix $L_m=L_m^\top\in\mathbb{R}^{n'\times n'}$ and a vector $\omega\in\mathbb{R}^{n'}$. Since $L_m$ is symmetric, it contains $n'(n'+1)/2$ independent elements. Define the parameter vector $l\in\mathbb{R}^{n'(n'+1)/2}$ by stacking the elements of the upper triangular part of $L_m$ row-wise, i.e.,
\begin{align*}
	l=
[&
		L_{m_{1,1}},
		L_{m_{1,2}},
		\cdots,
		L_{m_{1,n'}},
		L_{m_{2,2}},\\&
		L_{m_{2,3}},
		\cdots,
		L_{m_{2,n'}},
		\cdots,
		L_{m_{n',n'}}
	]^{\top}.
\end{align*}
Accordingly, a matrix $\Omega_m(\omega)\in\mathbb{R}^{n'\times n'(n'+1)/2}$ can be constructed such that
\begin{align}
	L_m\omega=\Omega_m(\omega)l.
	\label{Om_def}
\end{align}
In particular, for $\omega=[\omega_1,\ldots,\omega_{n'}]^\top$, the $i$th row of $\Omega_m(\omega)$ contains the elements $\omega_1,\ldots,\omega_{i-1}$ in the positions corresponding to $L_{m_{1,i}},\ldots,L_{m_{i-1,i}}$, followed by $\omega_i,\ldots,\omega_{n'}$ in the positions corresponding to $L_{m_{i,i}},\ldots,L_{m_{i,n'}}$, with zeros elsewhere. This parametrization is used throughout the paper for online estimation of the unknown matrix $L_m$. For example, for $n'=3$, it yields
\begin{align*}
\Omega_m=	\begin{bmatrix}
		\omega_1 & \omega_2 & \omega_3 & 0 & 0 & 0 \\
		0 & \omega_1 & 0 & \omega_2 & \omega_3 & 0 \\
		0 & 0 & \omega_1 & 0 & \omega_2 & \omega_3
	\end{bmatrix}
\end{align*}
\bibliographystyle{ieeetr}
	\bibliography{ref}
	
\end{document}